\documentclass{aims}
\usepackage{amsmath,amssymb,amsthm,caption}
  \usepackage{paralist}
  \usepackage{graphics} 
  \usepackage{epsfig} 
\usepackage{graphicx}  \usepackage{epstopdf}
 \usepackage[colorlinks=true]{hyperref}
\hypersetup{urlcolor=blue, citecolor=red}

\def\currentvolume{}

    \def\ppages{}

\newtheorem{theorem}{Theorem}[section]

\newtheorem{lemma}[theorem]{Lemma}

\theoremstyle{definition}

\newcommand{\eps}[1]{{#1}_{\varepsilon}}

\def\OO{\mathcal{O}}
\def\eps{\varepsilon}
\def\veps{\varepsilon}

\def\pa{\partial}

\def\R{\mathbb{R}}
\def\dd{\, \mathrm{d}}

\title[Justification of Quasicontinuum Models]
      {Justification of Leading Order Quasicontinuum Approximations of Strongly Nonlinear Lattices}

\author[Christopher Chong, P.G. Kevrekidis and Guido Schneider]{}



 \email{chong@math.umass.edu}
 \email{kevrekid@math.umass.edu}
 \email{guidos@mathematik.uni-stuttgart.de}

\begin{document}
\maketitle

\centerline{\scshape Christoper Chong and P.G. Kevrekidis  }
\medskip
{\footnotesize
 \centerline{Department of Mathematics and Statistics}
   \centerline{University of Massachusetts}
   \centerline{ Amherst, MA 01003-9305, USA}
} 

\medskip

\centerline{\scshape Guido Schneider}
\medskip
{\footnotesize
 \centerline{Institut f\"ur Analysis, Dynamik und Modellierung}
   \centerline{Universit\"at Stuttgart}
   \centerline{Pfaffenwaldring 57, 70569 Stuttgart, Germany.}
}

\bigskip


\begin{abstract}
We consider the leading order quasicontinuum limits of a one-dimensional 
granular medium governed by the Hertz contact law under precompression.
The approximate model which is derived in this limit is justified 
by establishing asymptotic bounds for the error with the help of energy estimates.
The continuum model predicts the development of shock waves, which are also studied 
in the full system with the aid of numerical simulations.
We also show that  existing results 
concerning the Nonlinear Schr\"odinger (NLS) and Korteweg de-Vries (KdV) approximation of FPU models apply directly to a precompressed granular medium in the weakly nonlinear regime. 
\end{abstract}


%
%
%

\section{Introduction}
\label{sec:intro}

We consider a one-dimensional granular medium which is governed by the Hertz law.
Denote $ q_n $ the relative displacement from equilibrium of the $ n $-th particle.
Then the renormalized equations of motion describing the $q_n$'s have the form \cite{Nesterenko2001},
\begin{equation} 
\ddot{q}_n  = W'(q_{n-1} - q_n)- W'(q_{n} - q_{n+1}) , \qquad (n \in \mathbb{Z}),
\end{equation}
where
\begin{equation}
W'(u) =   \left[ \delta_0 + u \right]_+^{p},  \qquad [u]_+ = u \, \Theta(u) =  \mathrm{max}(0,u)
\label{chain}
\end{equation}
where $\Theta$ is the Heaviside function and $\delta_0$ is the static load (precompression) applied to the chain at all times.
There is a double nonlinearity stemming from this system due to the lack of tensile strength 
(resulting in an asymmetric potential function $W$) and the nonlinear coupling. For spherical particles we have
$ p = 3/2$, which corresponds to the classical Hertz contact law. 

It will be convenient to work with the difference $u_n = q_{n-1}-q_n$, 
i.e. the strain, which satisfies 
\begin{equation} \label{law}
\ddot{u}_n = W'(u_{n+1})-2 W'(u_n)+ W'(u_{n-1}), \qquad (n \in \mathbb{Z}).
\end{equation}
With the long wave ansatz
\[   u_n(t) = A(X,T),  \qquad X=\varepsilon n, \quad T=\varepsilon t \]
where $ X,T, A(X,T) \in \mathbb{R}$ and $\eps \ll 1$ is a small perturbation parameter,  one can derive
the leading order continuum model approximation to the discrete dynamics 
\begin{equation} \label{eq:qc}
\partial^2_T A = \partial^2_{X} ((\delta_0+ A)^p).
\end{equation}
%
The main result of the present paper is the following approximation theorem.
\begin{theorem} \label{th1}
Fix $\delta_0 > 0$ and let $p\in\R^+$ and let 
$
A \in C \left(\left[ 0, T_0\right], H^4\right) $ with 
\[
\sup_{T \in \left[ 0, T_0\right]} \sup_{ X \in \mathbb{R}}
\left| A ( X, T) \right| \leq \delta_0/2
\qquad
\textrm{and}  \qquad  
\sup_{T \in \left[ 0, T_0\right]} \| A ( \cdot, T) \|_{H^4} \leq C_1 
\]
be a solution of \eqref{eq:qc}.
Then for $ C_1 > 0 $ sufficiently small there exists $C, \varepsilon_ 0 > 0$ such that for all
$\varepsilon \in (0, \varepsilon_0)$ there are solutions $(u_n(t))_{n
  \in \mathbb{Z}}$ of \eqref{law} satisfying
\[
\sup_{t\in [0, T_0/\varepsilon]} \sup_{n \in \mathbb{Z}} \left|
  u_n(t) - A(\varepsilon n, \varepsilon t)\right| < C \varepsilon^{3/2}.
\]
\end{theorem}


In the weakly nonlinear case, i.e. if the precompression is much greater than the amplitude of the solution,
then the Nonlinear Schr\"odinger (NLS) and Korteweg de-Vries (KdV) equations can be derived as continuum models
for solutions that have amplitude of order $\OO(\eps)$ and $\OO(\eps^2)$ respectively, see Section~\ref{sec:nls_kdv}.
For solutions of order $\OO(1)$ (i.e. solutions with  amplitude that are
on the same order as the precompression) model~\eqref{eq:qc} is relevant.
In order to establish local existence and uniqueness of the continuum model and to handle
the nonlinearity in Fourier space, we will expand the nonlinearity as a series.  Since the amplitude is $\OO(1)$
we cannot truncate this series (which is done when deriving the NLS and KdV equations). Therefore,
we require that the precompression is of the same order as the amplitude (but not greater). It is unclear if this is only a technical assumption
or if it is really necessary for having the approximation property. 
It will be the subject of future research to explore whether or not the assumptions of Theorem \ref{th1} can be
relaxed, and if so, how the continuum model \eqref{eq:qc} needs to be modified.

If one does not ignore higher order terms in the derivation (see next section), then one arrives
at the following $\OO(1)$ continuum model for the strain in
the purely nonlinear case (if $\delta_0=0$)
\begin{equation} 
\partial^2_T A = \partial^2_{X} (A^p)   + \frac{\epsilon^2}{12} \pa_x^4 (A^p) \label{eq:Pik} 
\end{equation}
which is the equation derived in \cite{pikovsky} (where the small parameter $\eps$ is formally set to unity).
From a physical standpoint, the reason for doing so is
understandable, as the higher order term in the model affords the existence of exact localized traveling wave solutions
(which do not exist in the continuum model~\eqref{eq:qc}). 
An alternative approach towards such a (higher order) quasi-continuum
model is the one spearheaded in the earlier work of~\cite{Nesterenko2001},
where one does a similar formal Taylor-expansion based calculation
at the displacement level and then differentiates with respect to $x$ to
derive the effective long-wavelength equation for the strains $r=q_x$.
Remarkably, it should be pointed out, as indicated in~\cite{pikovsky},
that these two procedures (reverting to strain variables and Taylor 
expanding to go to long wavelengths) do not commute, a feature which poses
a mathematical challenge in its own right (about their respective validity).
It should be noted that from a visual inspection the solitary wave profiles 
predicted
by~\eqref{eq:Pik} (and the corresponding ones of~\cite{Nesterenko2001}) 
compare quite well with the ``numerically exact'' traveling wave solution of the granular crystal model~\eqref{law} \cite{pikovsky}.
Nevertheless, without an approximation theorem, it is not clear 
to what level/extent these kinds of approximations can be used.
For example, for the initial value problem (and not just for very special
solutions, such as the exact traveling wave ones),
there are several nontrivial concerns about models such as those 
of~\cite{pikovsky,Nesterenko2001} in connection to their
 local uniqueness and existence properties
and the fact the solutions will depend on the small 
parameter $\epsilon$
(hence it is unclear what, if any, influence the $\OO(\epsilon)$ terms have on the dynamics on the $ O(1) $ time scale w.r.t. $ T $).
Thus, in terms of providing rigorous error estimates for the initial value problem, model~\eqref{eq:qc} is more appropriate, at least as a starting 
point or a leading order approximation.
Although Eq.~\eqref{eq:qc} fails to describe exact traveling solitary waves
(and hence the proper, controllable generalization to higher order so as to
capture this important trait is,
from a rigorous perspective, an open problem), the model does
well in other regards, such as in the description of shock wave formation. 
This, and other properties of~\eqref{eq:qc} are described in
Sections~\ref{sec:qc_exist} and \ref{sec:numerics}.

{\bf Notation.} Throughout this paper many possibly different constants are denoted with the 
same symbol $ C $ if they can be chosen independently of the small perturbation 
parameter $ 0 < \varepsilon \ll 1 $.

\section{Fourier transform as fundamental tool}
\label{sec:fourier}

The Fourier transform is the major tool 
in the proof of the approximation result. In this section we recall some basic facts and
establish notation conventions. 

Since the solutions $ u $ of  \eqref{law}  
live on $ \mathbb{Z} $ we need the 
Fourier transform on $ \mathbb{Z} $ leading to periodic functions 
in Fourier space.

\paragraph{\bf Fourier transform on $ \mathbb{Z} $:}
System \eqref{law} can be transferred into Fourier space by 
\begin{equation} \label{transform}
\hat{u}(k,t) = \mathcal{F}(u)(k,t) = \frac{1}{2\pi} 
\sum_{n \in \mathbb{Z}} u_n(t) e^{-ikn}.
\end{equation}
The inverse of $\mathcal{F} $ is given by 
\begin{equation} \label{transformback}
u_n(t) = (\mathcal{F}^{-1}\hat{u})_n(t) = 
\int_{-\pi}^{\pi} \hat{u}(k,t) e^{ikn} dk.  
\end{equation}
For every $ s \geq 0 $ the
Fourier transform $ \mathcal{F} $ is continuous from 
\[
\ell^{2}_s= \{ u:\mathbb{Z} \to \mathbb{R} \ | \ \| u(\cdot) \|_{\ell^{2}_s} <  \infty \}
\]
 into 

\begin{align*}
H^s_{\rm per}(\mathbb{R},\mathbb{C})  = & \{ \hat{u}: \mathbb{R} \to \mathbb{C} \ | \ 
 \hat{u}(\cdot) \ \mathrm{is} \  s \ 
\mathrm{times \ weakly \ differentiable}, \\
& \quad \| \hat{u}(\cdot) \|_{H^s_{\rm per}}< \infty, \  \hat{u}(k) =  \hat{u}(k+2\pi) \}
\end{align*}
where
\[
\| u(\cdot,t) \|_{\ell^{2}_s}^2 
= \sum_{n \in \mathbb{Z}} (|u_n(t)|^2 (1+n^2)^s)
\quad 
{\rm and}  
\quad 
\| \hat{u}(\cdot,t) \|_{H^s_{\rm per}}^2 = \sum_{m=0}^s (\int_{-\pi}^{\pi}
 |\partial_k^m 
\hat{u}(k,t) |^2 dk).
\]
The inverse Fourier transform $ \mathcal{F}^{-1} $ is continuous from 
$ H^s_{\rm per} $ into $ \ell^{2}_s $ and from $ L^1_{\rm per} $ 
into $ \ell^{\infty} $
where
\[
\| \hat{u}(\cdot,t) \|_{L^1_{\rm per}} =  \int_{-\pi}^{\pi} | 
\hat{u}(k,t) | dk
\qquad 
{\rm and}  
\qquad 
\| u(\cdot,t) \|_{\ell^{\infty}} 
= \sup_{n \in \mathbb{Z}} |u_n(t)|.
\]
\paragraph{\bf Fourier transform on $ \mathbb{R} $:}
Beside the Fourier transform on $ \mathbb{Z} $ we need the Fourier transform
on the real line in order to handle the continuum model  
which lives on the real line.
We set 
\begin{equation} \label{transformr}
\hat{u}(k,t) = \mathcal{F}(u)(k,t) = \frac{1}{2\pi} \int_{-\infty}^{\infty} u(x,t) e^{-ikx} dx.
\end{equation}
The inverse is given by 
\begin{equation} \label{transformbackr}
u(x,t) = \mathcal{F}^{-1}(\hat{u})(x,t) = 
\int_{-\infty}^{\infty} \hat{u}(k,t) e^{ikx} dk.  
\end{equation}
For every $ s \geq 0 $ the
Fourier transform $ \mathcal{F} $ is continuous from 
\[
L^{2}_s= \{ u:\mathbb{R} \to \mathbb{C} \ | \ \| u(\cdot) \|_{L^{2}_s} <  \infty \}
\] into 
\begin{eqnarray*}
H^s(\mathbb{R},\mathbb{C}) & = & \{ \hat{u}: \mathbb{R} \to \mathbb{C} \ | \ 
 \hat{u}(\cdot) \ \mathrm{is} \  s \ 
\mathrm{times \ weakly \ differentiable}, 
\ \| \hat{u}(\cdot) \|_{H^s}< \infty \   \}
\end{eqnarray*}
where
\[
\| u(\cdot,t) \|_{L^{2}_s}^2 
= \int_{-\infty}^{\infty} (|u(x,t)|^2 (1+x^2)^s) dx
\quad 
{\rm and}  
\quad 
\| \hat{u}(\cdot,t) \|_{H^s}^2 = \sum_{m=0}^s (\int_{-\infty}^{\infty} |\partial_k^m 
\hat{u}(k,t) |^2 dk).
\]
The Fourier transform $ \mathcal{F} $ and its inverse  $ \mathcal{F}^{-1} $ 
are continuous from 
$ H^s $ to $ L^{2}_s $ and vice versa. Moreover, they are continuous 
from $ L^1_s $ into $ C^s_b $
where
\[
\| \hat{u}(\cdot,t) \|_{L^1_s} =  \int_{-\infty}^{\infty} | 
\hat{u}(k,t) | (1+k^2)^{s/2} dk
\qquad 
{\rm and}  
\qquad 
\| u(\cdot,t) \|_{C^s_b} 
= \sum_{j=0}^s \sup_{x \in \mathbb{R}} |\partial_x^j u(x,t)|.
\]

\section{Derivation of the continuum model}
\label{sec:derive}


We assume the existence of a $\delta_0 > 0$ such that
$\inf_{n \in \mathbb{Z}} u_n(t) \geq - \delta_0$ for all $t \geq 0$. This is equivalent to requiring that the beads remain in contact for all times.
The existence of such solutions will be justified below. With this assumption we may ignore the Heaviside function in the potential, i.e. we 
have $ W'(u) = (\delta_0 + u)^p $.  Thus, we may expand the nonlinearity in a Taylor series,

$$(\delta_0 + u)^p = (\delta_0)^p (1+\frac{u}{\delta_0})^p = \sum^\infty_{\ell = 0} b_\ell u^\ell$$
with real-valued coefficients $b_{\ell} = b_{\ell} (\delta_0,p) $. With this expansion we find

\begin{equation} \label{law2}
\partial^2_t u_n = \sum^\infty_{\ell = 1} b_\ell (u^\ell_{n+1} - 2
u^\ell_n + u^\ell_{n-1}).
\end{equation}
Taking the 
Fourier-transform of the right hand side of \eqref{law2} yields

\begin{align*}
& \sum_{n \in \mathbb{Z}} \sum^\infty_{\ell = 1} b_\ell (u_{n+1}^{\ell} - 2 u_n^{\ell} + u_{n-1}^{\ell} ) e^{ikn}  = \sum_{n \in \mathbb{Z}} \sum^\infty_{\ell = 1} b_\ell (e^{ik} - 2 + e^{-ik}) u_{n}^\ell e^{ikn} \\
&  = - \omega (k)^2 \sum_{n \in \mathbb{Z}} \sum^\infty_{\ell = 1} b_\ell u^{\ell}_n e^{ikn}  = - \omega (k)^2 \sum^\infty_{\ell = 1} b_\ell \widehat{u}^{*\ell}(k)
\end{align*}
where $\widehat u^{\ast \ell}$ denotes the $(\ell-1)$-times convolution of
$\widehat u$ with itself, 
 $\omega (k)^2= 2(1-\cos (k))$ and $ \widehat u (k,t) = \widehat u (k+2\pi,t) $.
In Fourier space, the system \eqref{law} is therefore given by
\begin{equation} \label{Fourier}
\partial^2_t \widehat u (k,t) = - \omega (k)^2 \sum^\infty_{\ell = 1} b_\ell \widehat{u}^{*\ell}(k).
\end{equation}
We should note here in passing that the use of Fourier space techniques
in the context of granular systems was pioneered in the work of~\cite{pego05},
where it was used to develop an understanding of the decay properties
of traveling waves in the absence of precompression, as well as to 
offer an efficient numerical tool for computing them. Among the recent 
ramifications of this approach are the proof of the existence of such
bell-shaped traveling waves without~\cite{atanas1} and with~\cite{atanas2}
precompression.
 
The long wave limit ansatz
$ u(n,t) = A(\varepsilon n,\varepsilon t)  $
 is given in Fourier space by
\begin{equation} \label{eq45}
\widehat u (k,t) = \varepsilon^{-1} \widehat A \left(K, T\right),  \qquad  K = \frac{k}{\varepsilon}, \qquad T = \varepsilon t 
\end{equation}
with $ \widehat A : \mathbb{R} \to \mathbb{C} $ a function decaying to zero for $ |k| \to \infty $.
Inserting this ansatz into \eqref{Fourier}, rescaling the integrals, taking formally the limit
$\int_{-\pi/\veps}^{\pi/\veps} \rightarrow \int_{-\infty}^{\infty} $ yields
\[
\partial^2_T \widehat A ( K , T) = -  K^2  \sum^\infty_{\ell = 1} b_\ell  \widehat{A}^{*l} ( K , T) + \mathcal{O} (\varepsilon^2)
\]
where we used that $\omega (k)^2= k^2 + \mathcal{O}(k^4) = \eps^2 K^2 + \mathcal{O}((\veps K)^4) $.
Ignoring the higher order terms and taking the inverse Fourier transform 
of this expression yields our continuum  model  \eqref{eq:qc}. 
Note that keeping the next term in
the expansion of $\omega^2$ would yield Eq.~\eqref{eq:Pik} which, 
as mentioned above, is not covered by the proof presented herein.

Before we turn these formal calculations into rigorous arguments 
we consider properties of the continuum  model \eqref{eq:qc} itself.

\section{Local existence and uniqueness of the continuum equation} 
\label{sec:qc_exist}

We will need a certain regularity of the solutions of the continuum  model  \eqref{eq:qc},
therefore, we  prove the following existence and uniqueness result for the limit equation.

\begin{lemma}
Fix $\delta_0 > 0$ and $ s \geq 4$. Let $A_0 = A_0 (\cdot) \in H^s$ satisfy
$\sup_{{X} \in \mathbb{R}} \left|A_0 ( X) \right|
\leq \delta_0$ and $\left\| A_0 \right\|_{H^s} \leq C_1$ for a $C_1 >
0$.
Then for $ C_1 > 0 $ sufficiently small there exists a $T_0 > 0$ and solutions $A = A( X , T)$ of \eqref{eq:qc} with 
\begin{equation}\label{Lemma}
A\in C \left(
  \left[ 0,T_0\right], H^s\right) \textup{ and } \sup_{T\in \left[ 0,
    T_0\right]} \sup_{ X \in \mathbb{R}} \left| A( X,
  T)\right| \leq \delta_0/2.
\end{equation}
\end{lemma}
\begin{proof}

We know $A$  satisfies 
\[\partial^2_T A = \partial^2_{ X} \left( (\delta_0 + A)^p \right).\]
As before, we expand the nonlinearity in a series 

\begin{align*}
(\delta_0 + A)^p &= (\delta_0  )^p \left( 1 + \frac{A}{
    \delta_0}\right)^p = \sum^\infty_{\ell = 0} b_{\ell}A^{\ell} 
    \end{align*}
 with real-valued coefficients $b_{\ell}$ as defined before. The series 
is convergent for $\left| A\right| <   \delta_0 $.
Applying $ \partial^{-1}_{X} $ to \eqref{eq:qc} then multiplying the resulting equation with 
$\partial^{-1}_{ X} \partial_T A $ and then integrating w.r.t. $ X $ yields 

\begin{align*}
\int_{-\infty}^{\infty} \left( \partial^{-1}_{ X} \partial_T A \right)\left(  \partial^{-1}_{ X}  \partial^2_T A\right) \dd
 X 
&= \int_{-\infty}^{\infty} \left( \partial^{-1}_{ X} \partial_T
  A\right) \partial_X \left( (\delta_0 + A)^p\right) \dd 
X \\
&= \int_{-\infty}^{\infty} \left( \partial^{-1}_{ X} \partial_T A\right)
\sum^\infty_{\ell = 1} b_{\ell} \partial_X (A^{\ell}) \dd  X
\end{align*}
which yields

\begin{align*}
\frac12 \int_{-\infty}^{\infty} \partial_T \left( \left( \partial^{-1}_{
      X} \partial_T A\right)^2\right) \dd  X &= - \int_{-\infty}^{\infty}
\left( \partial_TA\right) \sum^\infty_{\ell= 1} b_{\ell} A^{\ell} \dd  X
\\
&= - \int_{-\infty}^{\infty} \sum^\infty_{\ell = 1} \frac{1}{\ell + 1} b_{\ell} \partial_T \left(
  A^{{\ell}+1}\right) \dd  X
\end{align*}
and so 

\[
 \partial_T E_0 = 0 
\]
where 

\[
E_0 = \left( \int_{-\infty}^{\infty} \frac12 \left( \partial^{-1}_{ X} \partial_T A
  \right)^2 + \sum^\infty_{\ell = 1} \frac{1}{\ell+1} b_{\ell} A^{{\ell}+1} \dd
   X\right).
\]
Proceeding similarly for the derivatives yields

\begin{align*}
\int_{-\infty}^{\infty} \left( \partial^s_{ X} \partial_T A
\right) \left( \partial^s_{ X} \partial^2_T A \right)\dd  X
&= \int_{-\infty}^{\infty} \left( \partial^s_{ X} \partial_T A
\right) \partial^s_{ X} \partial^2_X \left( (\delta_0 +
  A)^p\right) \dd  X \\
&= \int_{-\infty}^{\infty} \left( \partial^s_{ X} \partial_T A\right)
\sum^\infty_{\ell = 1} b_{\ell} \partial^{s+2}_X (A^{\ell})\dd  X.
\end{align*}
By partial integration we obtain 

\begin{align*}
\frac12 \int_{-\infty}^{\infty} \partial_T \left( \partial^s_{ X} \partial_T
  A\right)^2 \dd  X &= - \int_{-\infty}^{\infty}
\left( \partial^{s+1}_{ X} \partial_T A \right)
\sum^\infty_{\ell = 1} b_{\ell} \partial^{s+1}_{ X} (A^{\ell}) \dd
 X \\
&= - \frac12 \int_{-\infty}^{\infty} \sum^\infty_{\ell = 1} \ell b_{\ell} A^{{\ell}-1} \partial_T \left(
  \left( \partial^{s+1}_{ X} A\right)^2\right) \dd
 X + \widetilde G_s \\
&= - \frac12 \int_{-\infty}^{\infty} \partial_T \left(\sum^\infty_{\ell = 1} \ell b_{\ell} A^{{\ell}-1}
  \left( \partial^{s+1}_{ X} A\right)^2 \right) \dd
 X + G_s
\end{align*}
where $\widetilde G_s, G_s$ only contain terms with at
most $s+1$ spatial and temporal derivatives.
Therefore

\[
 \partial_T E_{s+1} = G_s
\]
where
\[
E_{s+1} = \frac12 \int_{-\infty}^{\infty} \left( \partial^s_{ X} \partial_T A\right)^2  +
\left( \sum^\infty_{\ell = 1} \ell b_{\ell} A^{{\ell}-1}
  \left( \partial^{s+1}_{ X} A\right)^2 \right) \dd
 X
\]
and
\[
\left| G_s\right| \leq C \left( \left\| A\right\|^3_{H^{s+1}} +
  \left\| \partial_T A \right\|^3_{H^s}\right)
\]
if $\left\| A\right\|_{H^{s+1}}$ and $\left\| \partial_T
  A\right\|_{H^s}$ are smaller than half of the radius of convergence of
the involved series. Note that all series have the same radius of
convergence, namely $ \delta_0 $.

Since a multiple of $\mathcal{E}_{s+1} = E_0 + \ldots + E_{s+1}$ is an upper bound of  the squared
$H^{s+1}$-norm we obtain an estimate 
\[
 \partial_T \mathcal{E}_{s+1} \leq C \mathcal{E}_{s+1}^{3/2}.
\]
Hence by Gronwall's inequality we can guarantee that
$\left\| A\right\|_{H^{s+1}}$ 
stays in between half the radius of convergence for all $t\in \left[0,
  T_0\right]$ if $T_0 > 0$ and $ C_1 > 0 $ are chosen sufficiently small.
Since for $x \in \mathbb{R}$ the sup-norm can be estimated by the
$H^1$-norm, the second inequality of \eqref{Lemma} follows too.
Since the previous a priori estimates guarantee that we have a quasilinear system in the sense 
of \cite{Kato75}, the local existence and uniqueness of solutions follows.
\end{proof}

\section{Estimates for the residual}
\label{sec:residual}

For the proof of the approximation result we need a way to measure how ansatz \eqref{eq45}
fails to satisfy \eqref{law}, i.e. we will need estimates for the residual.

It turns out to be advantageous to work in Fourier space, i.e., to work with 
\eqref{Fourier} instead of \eqref{law}. The error $ \varepsilon^{\beta} \widehat{R} = \widehat{u} -\widehat{\mathcal{A}} $ 
with $ \widehat{\mathcal{A}}(k,t) = \varepsilon^{-1} \widehat{A}(K,T) $
satisfies 

\begin{align*} 
 \varepsilon^{\beta} \partial^2_t \widehat R (k,t)  
= & - \omega (k)^2   \sum^\infty_{\ell = 1}  b_{\ell}  \widehat{(\mathcal{A}+\varepsilon^{\beta} R)^{*\ell}}(k,t) - \partial^2_t \widehat{\mathcal{A}} (k,t) \\
=&- \omega (k)^2   \sum^\infty_{\ell = 1}  b_{\ell}  \widehat{(\mathcal{A}+\varepsilon^{\beta} R)^{*\ell}}(k,t)  
 +  \omega (k)^2 \sum^\infty_{\ell = 1}  b_{\ell}  \widehat{\mathcal{A}}^{*\ell}(k,t) \\
& - \omega (k)^2 \sum^\infty_{\ell = 1}  b_{\ell} \widehat{\mathcal{A}}^{*\ell}(k,t) + k^2 \sum^\infty_{\ell = 1}  b_{\ell} \widehat{\mathcal{A}}^{*\ell}(k,t)
\end{align*}
and so 

\begin{equation} \label{error}
 \varepsilon^{\beta} \partial^2_t \widehat R (k,t)  = - \omega (k)^2 \sum^\infty_{\ell = 1}  b_{\ell}  \widehat{(\mathcal{A}+\varepsilon^{\beta} R)^{*\ell}}(k,t)
+ \omega (k)^2 \sum^\infty_{\ell = 1}  b_{\ell}  \widehat{\mathcal{A}}^{*\ell}(k,t)
+ \widehat{\mathop{\mathrm{Res}} (\mathcal{A})}(k,t) 
\end{equation}
where 

\[
\widehat{\mathop{\mathrm{Res}} (\mathcal{A})}(k,t) = \left( k^2 - \omega (k)^2 \right)\sum^\infty_{\ell = 1}  b_{\ell} \widehat{\mathcal{A}}^{*\ell}(k,t)
\]
stands for  the residual terms,
i.e., for the terms which we neglected within the leading order
 quasicontinuum  approximation.
In order to solve \eqref{error} with boundary conditions $ \widehat R (k,t)  = \widehat R (k+2\pi ,t) $ 
we have to modify $ \widehat{\mathcal{A}}(k,t)$ which decays to zero for $ |k | \to \infty $.
We multiply $\widehat{\mathcal{A}}(k,t)$ with a cut-off function $ \chi_{[-\pi/2,\pi/2]} $. The
segment from $[-\pi,\pi]$ is then extended periodically with a period $ 2 \pi $
to the entire real axis. Call the outcome $ \widehat{\mathcal{A}}_{\#}(k,t)$. 
In exactly the same way we modify $-k^2 $ in the residual, which is then denoted with $ -k_{\#}^2 $.

We will need an estimate of the difference  $ \widehat{\mathcal{A}}(k,t)-\widehat{\mathcal{A}}_{\#}(k,t) $ and the
error $ \varepsilon^{\beta} \widehat{R} = \widehat{u} -\widehat{\mathcal{A}_{\#}} $ which satisfies 
\begin{eqnarray} \label{nerror}
 \varepsilon^{\beta} \partial^2_t \widehat R (k,t)  = - \omega (k)^2 \sum^\infty_{\ell = 1}  b_{\ell}  \widehat{(\mathcal{A_{\#}}+\varepsilon^{\beta} R)^{*\ell}}(k,t)
+ \omega (k)^2 \sum^\infty_{\ell = 1}  b_{\ell}  \widehat{\mathcal{A}}^{*\ell}_{\#}(k,t)
+ \widehat{\mathop{\mathrm{Res}}_{\#} (\mathcal{A_{\#}})}(k,t)   \nonumber
\end{eqnarray}
where 
\[
\widehat{\mathop{\mathrm{Res}}_{\#}  (\mathcal{A_{\#}})}(k,t) = \left( k_{\#} ^2 - \omega (k)^2 \right)\sum^\infty_{\ell = 1}  b_{\ell} \widehat{\mathcal{A}}^{*\ell}_{\#} (k,t).
\]

Thus, in order to bound the error $R$, we will need estimates for the residual $ \displaystyle \mathop{\mathrm{Res}}_{\#}(\mathcal{A}_{\#} )$.
Since $L^2$ is closed under convolution  on $[-\pi,\pi]$
it turns out to be sufficient to make the estimates in $ L^2 $.
We have
\begin{lemma} \label{lemn3}
Fix $ \delta_0 > 0 $. Let
$
A \in C \left(\left[ 0, T_0\right], H^4\right) $
with $  
\sup_{T \in \left[ 0, T_0\right]} \sup_{ X \in \mathbb{R}}
\left| A ( X, T) \right| \leq \delta_0/2 $
be a solution \eqref{eq:qc}.
Then there exist $\varepsilon_0, C > 0$ such that for all $\varepsilon
\in (0, \varepsilon_0)$ we have
\[
\sup_{t\in [0, T_0/\varepsilon]} \left\|\widehat{\mathop{\mathrm{Res}}_{\#} (\mathcal{A}_{\#} )}
\right\|_{L^2_{\mathrm{per}}} < C \varepsilon^{7/2}
\]
and
\[
\sup_{t\in [0, T_0/\varepsilon]} \left\|\omega^{-1} \widehat{\mathop{\mathrm{Res}}_{\#} (\mathcal{A}_{\#} )}
\right\|_{L^2_{\mathrm{per}}} < C \varepsilon^{5/2}.
\]

\end{lemma}
\begin{proof}
For completeness we recall the proof of \cite[Lemma 3.3]{CCPS12}.
We have that
$ \widehat{\mathcal{G}} = \sum^\infty_{\ell = 2}  b_{\ell} \widehat{\mathcal{A}}^{*\ell}_{\#} (k,t)    $   
satisfies 
\[
\sup_{t \in [0,T_0/\varepsilon]} \left(\int_{-\pi}^{\pi} |  \widehat{\mathcal{G}}(k,t) |^2 (1+ \frac{k^2}{\varepsilon^2})^4 dk\right)^{1/2} = \mathcal{O}(\varepsilon^{-1/2})
\] 
for $ \varepsilon \to 0 $, (note the index of $\mathcal{G}$ starts at $\ell=2$).  The loss of $\varepsilon^{-1/2}$ comes from the scaling properties of the $L^2$-norm.
Using that $k^2 - \omega (k)^2 = \mathcal{O}(k^4) $ then yields 
\begin{eqnarray*}
&& \left( \int_{-\pi}^{\pi} | (k^2 - \omega (k)^2) \widehat{\mathcal{G}}(k,t) |^2  dk \right)^{1/2}  \leq 
 C \left(\int_{-\pi}^{\pi} | k^4 \widehat{\mathcal{G}}(k,t) |^2  dk \right)^{1/2} \\ & \leq &
C  \sup_{k \in [-\pi,\pi]} \left| \frac{k^4}{(1+ \frac{k^2}{\varepsilon^2})^2} \right|    \left(\int_{-\pi}^{\pi} |  \widehat{\mathcal{G}}(k,t) |^2 (1+ \frac{k^2}{\varepsilon^2})^4 dk \right)^{1/2} \leq C \varepsilon^{7/2}.
\end{eqnarray*}
Since $b_1$ is independent of $\veps$ and since $k^2 - \omega (k)^2 = \mathcal{O}(k^4) $, it follows that

\begin{equation*}
 \left( \int_{-\pi}^{\pi} \left| (k^2 - \omega (k)^2) \sum^\infty_{\ell = 1}  b_{\ell} \widehat{\mathcal{A}}^{*\ell}_{\#} (k,t)  \right|^2  dk \right)^{1/2}  \leq C \varepsilon^{7/2}.
\end{equation*}
Applying the same argument as above with $ (k^2 - \omega (k)^2)/\omega(k) = \OO(k^3)$ yields

\begin{equation*}
 \left( \int_{-\pi}^{\pi} \left| \omega^{-1}(k)(k^2 - \omega (k)^2) \sum^\infty_{\ell = 1}  b_{\ell} \widehat{\mathcal{A}}^{*\ell}_{\#} (k,t)  \right|^2  dk \right)^{1/2}  \leq C \varepsilon^{5/2}.
\end{equation*}
\end{proof}
We close this section with an estimate for the difference  $ \widehat{\mathcal{A}}(k,t)-\widehat{\mathcal{A}}_{\#}(k,t) $.
\begin{lemma} \label{Lemma4}
Let $ A \in C([0,T_0],H^4) $ then 
\[
\sup_{t \in [0,T_0/\varepsilon]} \sup_{n \in \mathbb{Z}} | \mathcal{A}(n,t) - {\mathcal{A}}_{\#}(n,t) | \leq C \varepsilon^{7/2}.
\]
\end{lemma}
\begin{proof}
We have 
\[
\mathcal{A}(n,t) = \int_{-\infty}^{\infty} \varepsilon^{-1}\widehat{A}(\frac{k}{\varepsilon}, \varepsilon t) e^{ikn} dk
\qquad
\textrm{and} 
\qquad
\mathcal{A}_{\#}(n,t) = \int_{-\pi/2}^{\pi/2} \varepsilon^{-1}\widehat{A}(\frac{k}{\varepsilon}, \varepsilon t) e^{ikn} dk.
\]
For the difference we obtain
\begin{eqnarray*}
&& \left| \int_{\pi/2}^{\infty} \varepsilon^{-1}\widehat{A}(\frac{k}{\varepsilon}, \varepsilon t) e^{ikn} dk \right| \\
& \leq & C  \sup_{k \in [\pi/2,\infty]} \left| \frac{1}{(1+ \frac{k^2}{\varepsilon^2})^2} \right|  \left(\int_{-\infty}^{\infty} |  \varepsilon^{-1}\widehat{A}(\frac{k}{\varepsilon}, \varepsilon t) |^2 (1+ \frac{k^2}{\varepsilon^2})^4 dk \right)^{1/2} \leq C \varepsilon^{7/2}
\end{eqnarray*}
uniformly in $ n $. The loss of $\varepsilon^{-1/2}$ again comes from the scaling properties of the $L^2$-norm.

\end{proof}

\section{The error estimates}
\label{sec:error}

It remains to bound the solutions of \eqref{nerror}. In  accordance with Lemma \ref{lemn3} we choose $\beta = 3/2$.
 We proceed as in Section \ref{sec:qc_exist}  using energy
estimates.

From Section~\ref{sec:residual} we know

\[
\partial^2_t {\widehat{R}} = - \omega^2 \cdot \left( \sum^\infty_{\ell = 0}
  b_{\ell} (\left(  {\widehat{\mathcal{A}_{\#}}} + \varepsilon^\beta {\widehat{R}}\right)^{*\ell} - {\widehat{\mathcal{A}_{\#}}}^{*\ell})\right)\varepsilon^{-\beta } +
  \varepsilon^{-\beta } \widehat{\mathop{\mathrm{Res}}_{\#}(\mathcal{A}_{\#})}.
\]
Define the energy

\[
E_0 =  \int_{-\pi}^{\pi} |\omega^{-1} \partial_t {\widehat{R}}|^2 dk/2  + \int_{-\pi}^{\pi} |\omega^{-1} b_1 {\widehat{R}}|^2 dk/2.
\]
Making use of the fact

\[
\textrm{Re }\int_{-\pi}^{\pi} \omega^{-2} \overline{\partial_t {\widehat{R}}} \partial^2_t {\widehat{R}} dk = \partial_t \int_{-\pi}^{\pi}
|\omega^{-1} \partial_t {\widehat{R}}|^2 dk/2.
\]
we can compute

\begin{align*}
\pa_t E_0  = & \mathrm{Re} \left\{ - \varepsilon^{-\beta } \int_{-\pi}^{\pi} (\overline{\partial_t
{\widehat{R}}}) \left( \sum^\infty_{\ell = 2} b_{\ell}  (\left(  {\widehat{\mathcal{A}_{\#}}} + 
\varepsilon^\beta {\widehat{R}}\right)^{*\ell} - {\widehat{\mathcal{A}_{\#}}}^{*\ell})\right) dk  \right.
 \\
&  \left.  \vphantom{ \sum^\infty_{\ell = 0}}  + \varepsilon^{-\beta} \int_{-\pi}^{\pi}
\left( \omega^{-1} \overline{\partial_t {\widehat{R}}} \right) \left( \omega^{-1}\widehat{\mathop{\mathrm{Res}}_{\#}(\mathcal{A}_{\#})}
\right)dk. \right\}
\end{align*}
Note that the autonomous linear terms have canceled, explaining why the sum begins at $\ell=2$. Recall that the application of $ \omega^{-1} $ to the residual terms is well defined, see Lemma~\ref{lemn3}. We show below (see e.g.~\eqref{CE})
that the application of $ \omega^{-1} $ to $\pa_t \hat{R}$ and $\hat{R}$ is also well defined.
%
%
Using the Plancherel's identity allows us to rewrite 

\[
\varepsilon^{-\beta } \textrm{Re } \int_{-\pi}^{\pi} (\overline{\partial_t
{\widehat{R}}}) \left( \sum^\infty_{\ell = 2} b_{\ell}  (\left(  {\widehat{\mathcal{A}_{\#}}} + \varepsilon^\beta {\widehat{R}}\right)^{*\ell} - {\widehat{\mathcal{A}_{\#}}}^{*\ell})\right) dk
\] 
as 

\begin{eqnarray*}
&& \varepsilon^{-\beta } \sum_{n \in \mathbb{Z}} (\partial_t {{R_n}}) \left( \sum^\infty_{\ell = 2} b_{\ell} \left( \left(
      A_n+\varepsilon^\beta {{R_n}}\right)^{\ell} - A_n^{\ell}\right) \right) \\
&= & \varepsilon^{-\beta } \sum_{n \in \mathbb{Z}}\sum^\infty_{\ell = 2} b_{\ell} (\partial_t {{R_n}}) \sum^{\ell}_{j = 1}  
\begin{pmatrix} l \\ j \end{pmatrix}{ (\veps^{\beta} R_n   })^j A_n^{{\ell}-j} \\
&= &\sum_{n \in \mathbb{Z}}\sum^\infty_{\ell = 2} b_{\ell}
\sum^{\ell}_{j = 1}
\begin{pmatrix} l \\ j \end{pmatrix}
\frac{1}{j+1} \partial_t \left( {{R_n}}^{j+1}\right) \veps^{\beta(j-1)} A_n^{{\ell}-j} \\
&= &\partial_t \sum_{n \in \mathbb{Z}}\left( \sum^\infty_{\ell = 2} b_{\ell} \sum^{\ell}_{j = 1}
\begin{pmatrix} l \\ j \end{pmatrix} \frac{1}{j+1} ({{R_n}}^{j+1})
A_n^{{\ell}-j}\veps^{\beta(j-1)} \right) \\
&&-\sum_{n \in \mathbb{Z}}\sum^\infty_{\ell = 2} b_{\ell} \sum^{\ell}_{j = 1} 
\begin{pmatrix} l \\ j \end{pmatrix}
\frac{1}{j+1} ({{R_n}}^{j+1})\partial_t (A_n^{{\ell}-j})\veps^{\beta(j-1)} ,
\end{eqnarray*}
where $A_n$ (resp. $R_n$) is the inverse discrete Fourier transform of
$\widehat{\mathcal{A}_{\#}}$ (resp. $\widehat{R}$) evaluated at $n$.  This motivates the definition of a modified energy

\[
E_1 = E_0 + \sum_{n \in \mathbb{Z}}\left( \sum^\infty_{\ell = 2} b_{\ell} \sum^{\ell}_{j = 1}
\begin{pmatrix} l \\ j \end{pmatrix} \frac{1}{j+1} ({{R_n}}^{j+1})
A_n^{{\ell}-j} \veps^{\beta(j-1)} \right)
\]
which by construction satisfies,
\[
\partial_t E_1 = H_0 +  \varepsilon^{-\beta} \int_{-\pi}^{\pi}
\left( \omega^{-1} \overline{\partial_t {\widehat{R}}} \right) \left( \omega^{-1}\widehat{\mathop{\mathrm{Res}}_{\#}(\mathcal{A}_{\#})}
\right)dk,
\]
where
\[
H_0 = \sum_{n \in \mathbb{Z}}\sum^\infty_{\ell = 2} b_{\ell} \sum^{\ell}_{j = 1} 
\begin{pmatrix} l \\ j \end{pmatrix}
\frac{1}{j+1} ({{R_n}}^{j+1})\partial_t (A_n^{{\ell}-j})\veps^{\beta(j-1)}.
\]
%
Let  $\| A(\cdot,T)  \|_{\ell^2} \leq C_2$ and  $\| \pa_t A(\cdot,T)  \|_{\ell^2} \leq C_3 \veps $ and
$\| R(\cdot,t) \|_{\ell^2} \leq  C_E$ with $C_E$ defined below in~\eqref{CE}. Then

\begin{align*} \displaystyle
H_0  \leq  & \| R \|_{\ell^2}  \sum^\infty_{\ell = 2} b_{\ell} \sum^{\ell}_{j = 1} 
\begin{pmatrix} l \\ j \end{pmatrix}
\frac{1}{j+1}  \| R \|^j _{\ell^2}       \| \partial_t (A^{{\ell}-j} ) \|_{\ell^2} \veps^{\beta(j-1)}  \\
& \leq  \| R\|_{\ell^2} \left( \veps  \| R \| _{\ell^2} C_4  + \veps^{\beta + 1} \| R \|^2 _{\ell^2}  C_5(C_E)     \right)
\end{align*}
where 

\begin{align*}
C_4 &:=  \sum^\infty_{\ell = 2} b_{\ell}
\frac{1}{2}     (l-1) C_2^{{\ell}-2}  C_3,      
 \\   
C_5(C_E) &:=  \sum^\infty_{\ell = 2} b_{\ell} \sum^{\ell}_{j = 2} 
\begin{pmatrix} l \\ j \end{pmatrix}
\frac{1}{j+1}   C_E^{j-2}  (l-j) C_2^{{\ell}-j-1}  C_3   \veps^{\beta(j-2)}.     
\end{align*}
The first series is convergent for $C_1$ (and hence $C_2$) sufficiently small.
By changing indices, factorizing and ignoring a finite number of terms in the second series 
results in the expression $C_5(C_E) \approx \delta_0^p( 1 + ( C_2 + \veps^\beta C_E)/( \delta_0))^p$
and thus the series is convergent
if $C_1$ and $\veps$ are sufficiently small. 
Due to the Plancherel identity 
the energy $E_1$ is an upper bound for the squared $L^2$-norm for $ C_1 >
0$ sufficiently small.
Then since $\omega$ is bounded and since $L^2\subset L^1$ on bounded
domains (as a consequence $ L^2 $ is closed under convolution),  we find

\begin{align*}
H_0  \leq   ( C_4 + 1) E_1 \veps
\end{align*}
where we pick $\veps>0$ small enough such that $ \veps^\beta C_5(C_E) C_E \leq 1$. Thus
\[
\partial_t E_1 \leq ( C_4 + 1 + C_6) E_1 \veps + C_6  \veps 
\]
with constants $ C_j $ independent of $ \varepsilon $. For the inequality above we used the fact that
\[
\left|\varepsilon^{-\beta} \int_{-\pi}^{\pi}
\left( \omega^{-1} \overline{\partial_t {\widehat{R}}} \right) \left( \omega^{-1}\widehat{\mathop{\mathrm{Res}}_{\#}(\mathcal{A}_{\#})}
\right)dk \right|
\leq C_6 \varepsilon
\left( 1 +  \int_{-\pi}^{\pi}
|\omega^{-1} \partial_t {\widehat{R}}|^2 dk
\right),
\]
which is a consequence of the Cauchy-Schwarz inequality, Lemma~\ref{lemn3} and the fact that $\sqrt{x} < 1 + x$.
Hence by Gronwall's inequality we have
\begin{equation} \label{CE}
\sup_{t\in \left[ 0,T_0/\varepsilon\right]} E_1 (t) 
\leq \frac{C_6}{C_4+1+C_6} e^{(C_4+1+C_6)T_0} - \frac{C_6}{C_4+1+C_6}  = : C_E^2
 \end{equation}
for all $ t \in [0,T_0/\veps]$ where we set $u(\cdot,0)=A(\veps \, \cdot,0)$ such that $E_1(0)=0$. 
As a direct consequence we have   
  \[
\sup_{t\in [0, T_0/\varepsilon]} \sup_{n \in \mathbb{Z}} \left|
  u_n(t) - \mathcal{A}_{\#}(n, t)\right| < C \varepsilon^{3/2}.
\]
Combining this estimate with the one from Lemma \ref{Lemma4} gives the assertion of 
Theorem \ref{th1}, completing our proof.


\qed 

\section{Shock Formation in Granular Media}
\label{sec:numerics}

\begin{figure} 
  \centerline{ 
\epsfig{file=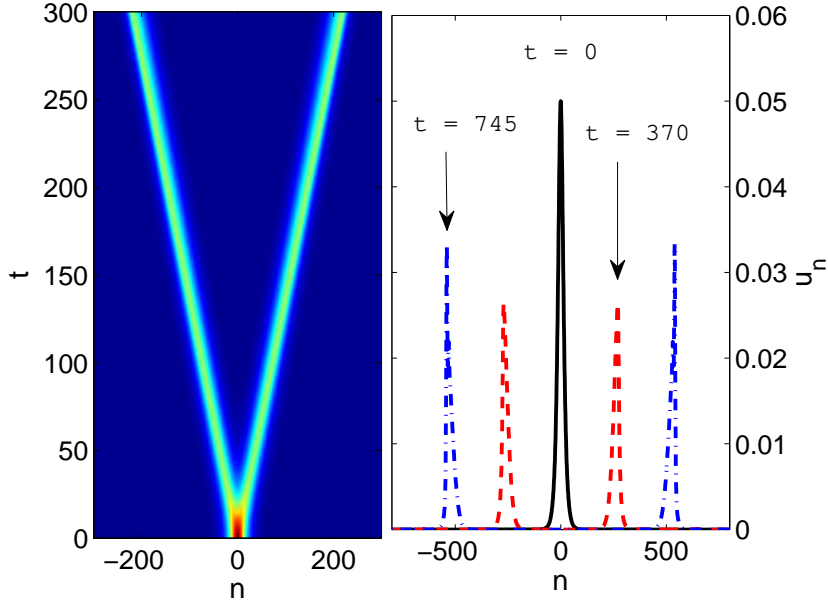,width=.9 \textwidth}
}
\caption{Development of a shock-like structure with a smooth and localized initial condition of the granular crystal model~\eqref{law}
with $p=3/2$ and $\delta = 0.1$. A space-time contour plot  is shown on the left and spatial profiles  at time $t=0$ (solid black line),
$t=370$ (red dashed line) and $t=750$ (blue dashed-dot line) are shown in the right panel.
In both panels it can be seen that the pulse separates into two counter-propagating waves. A zoom of the wave before and after the point of shock development is shown 
in Fig.~\ref{fig:oscillate}. }
\label{fig:shock}
\end{figure}


We now briefly discuss some analytical and numerical observations
that stem from the leading order quasicontinuum approximation developed
in Theorem \ref{th1}.
The continuum model

\begin{equation*} 
\partial^2_T A = \partial^2_{X} ((\delta_0+ A)^p)
\end{equation*}
can be written as a system of conservation laws

\begin{equation} \label{psystem}
\left .
\begin{array}{cc}
 \pa_T A - \pa_X v &= 0 \\
\pa_T v - \pa_X [  (\delta_0 + A)^p ] &=0
\end{array}
\right\}
\end{equation}
These conservation laws have the form of a so-called p-system~\cite{LeVeque,Smoller}.
Defining 

$$U = \begin{pmatrix}  A \\ v  \end{pmatrix},  \qquad  F(U) = \begin{pmatrix}  -v \\ -(\delta_0 + A)^p   \end{pmatrix}$$ 
we can express ~\eqref{psystem} as,

$$ \pa_T U + \pa_X F(U) = 0$$
or equivalently as 

\begin{equation} \label{advection}
 \pa_T U +  DF(U) \pa_X U = 0 
\end{equation}
where 
$$ DF(U) = \begin{pmatrix}
    0  & -1 \\ -p( \delta_0 + A)^{p-1} &0 
   \end{pmatrix} 
   $$
which has the eigenvalues $ \lambda_{\pm}(U) = \pm \sqrt{ p( \delta_0 + A)^{p-1}   } $. Thus, solutions
of \eqref{eq:qc} will consist of two counter-propagating waves traveling with velocity $\pm \sqrt{ p( \delta_0 + A)^{p-1}   } $.
Since the wave speed will depend on the amplitude of the solution, 
bell-shaped initial data will deform and steepen and 
a shock wave will form in finite time.
Discrete shock-like structures (which we will simply call shock waves) have been studied in granular media, e.g. in 
homogeneous and periodic chains with $p=3/2$
and $\delta_0 = 0$ \cite{HerboldPRE07,Molinari2009}. In those works, however, the
shock wave is generated by applying a velocity to a single bead \cite{HerboldPRE07} or by imparting velocity
to the end of the chain continuously \cite{Molinari2009}. In this paper, the mechanism for the development of the shock wave is fundamentally different.
It manifests from an arbitrary non-monotonically-increasing initial 
strain profile under precompression and given a sufficiently long time to develop.

A natural question is if the shock wave formation predicted by the continuum model~\eqref{eq:qc} is also
present in full system~\eqref{law}. Theorem~\ref{th1} no longer applies in this case, as the
the shock wave violates the required smoothness condition. Nonetheless, we 
carry out a numerical simulation to address the relevant
question with a smooth initial condition. 
Figure~\ref{fig:shock} shows the development of a shock wave in the discrete
model~\eqref{law} with the initial condition 
$u_n(0) = a \, \mathrm{sech}( b \epsilon n), v_n(0) = 0$ where $a,b\in\R$ are shape parameters.
The wave propagation closely follows the theoretical expectation on the basis 
of the leading order quasicontinuum approximation. Moreover, the predicted velocity
$\pm \sqrt{ p( \delta_0 + A)^{p-1}   } $ proved to be very accurate (for example, there was less than a $\%0.01$ relative error
in the case shown in Fig.~\ref{fig:shock}).

We would like to make a direct comparison of the solutions of the continuum model and the discrete model. However, one has to be careful
when using numerical approximations of the continuum model. For example, if one uses a finite difference
 approximation for a spatial discretization of~\eqref{eq:qc}, then we arrive at a model identical to the granular model~\eqref{law}.
One could use other numerical schemes, such as those based on adding artificial dispersion \cite{LeVeque}, but we will
proceed in an alternative way to predict the development of a shock wave.  Using the velocity relationship $\pm \sqrt{ p( \delta_0 + A)^{p-1}   }$, we can construct the profile of the continuum model
for an arbitrary time, starting with the left or right wave  (once they are separated) as an initial profile, see Fig.~\ref{fig:oscillate}.
A numerical computation is used to separate the profiles (i.e. we simulate~\eqref{eq:qc} with a finite difference method until separation but before the development of
any shock wave, thus avoiding any issues with smoothness). The continuum model predicts a shock wave
for any time past the point of non-single-valuedness, as shown in right panel Fig.~\ref{fig:oscillate}.
%

In FPU lattices, it is well known that dispersive shocks can develop, in which microscopic oscillations spread out in
space and time \cite{Herrmann10}. From Fig.~\ref{fig:oscillate} one clearly sees 
near the point of wave breaking the development of such oscillations, which
are absent in the continuum model \cite{Herrmann10}. Thus, it would be 
relevant to extend works like \cite{Herrmann10} in order to better understand shock waves
in the granular crystal model~\eqref{law}.


\begin{figure} 
  \centerline{ 
\epsfig{file=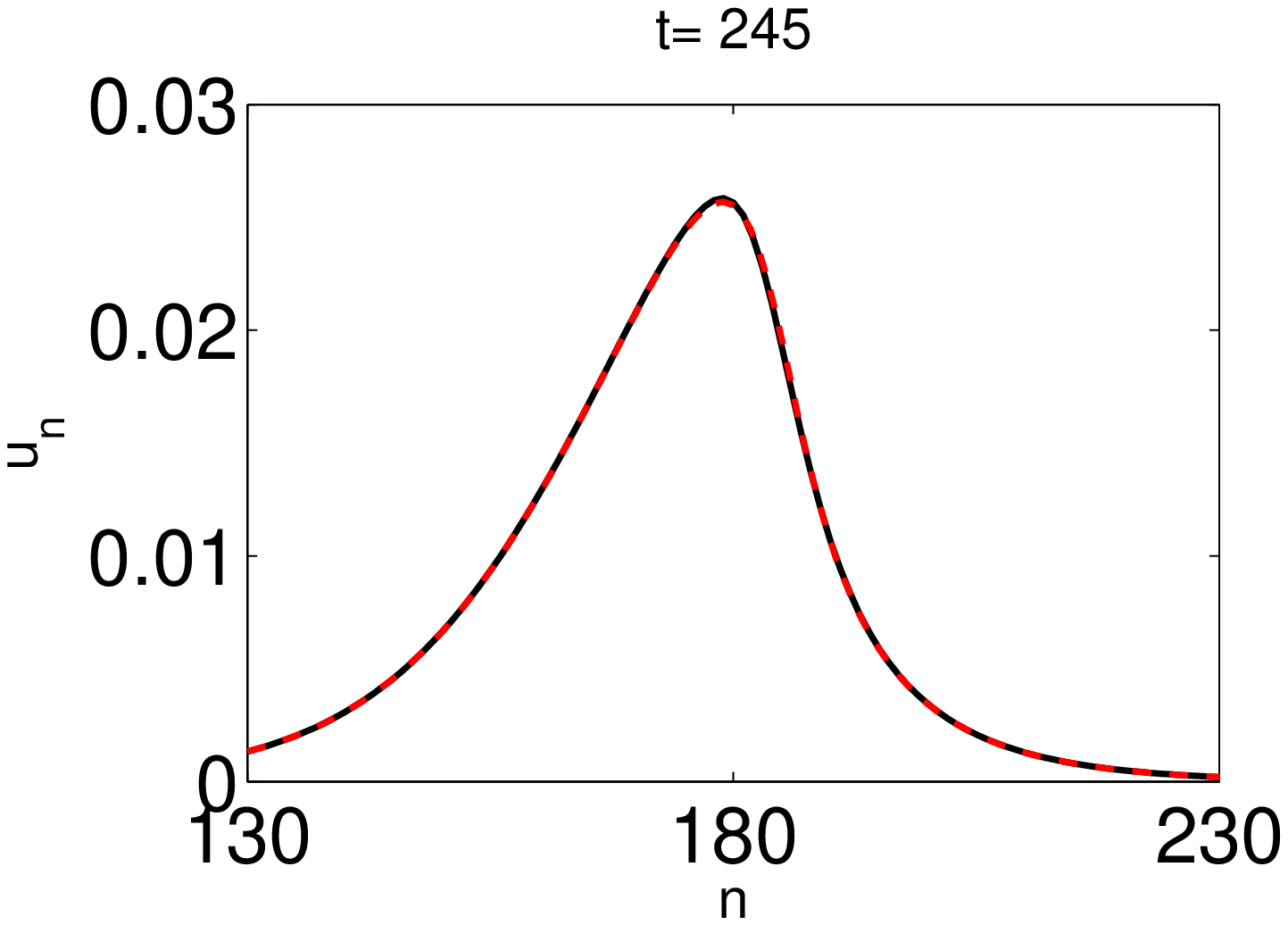,width=.33 \textwidth, height = .15 \textheight}
\epsfig{file=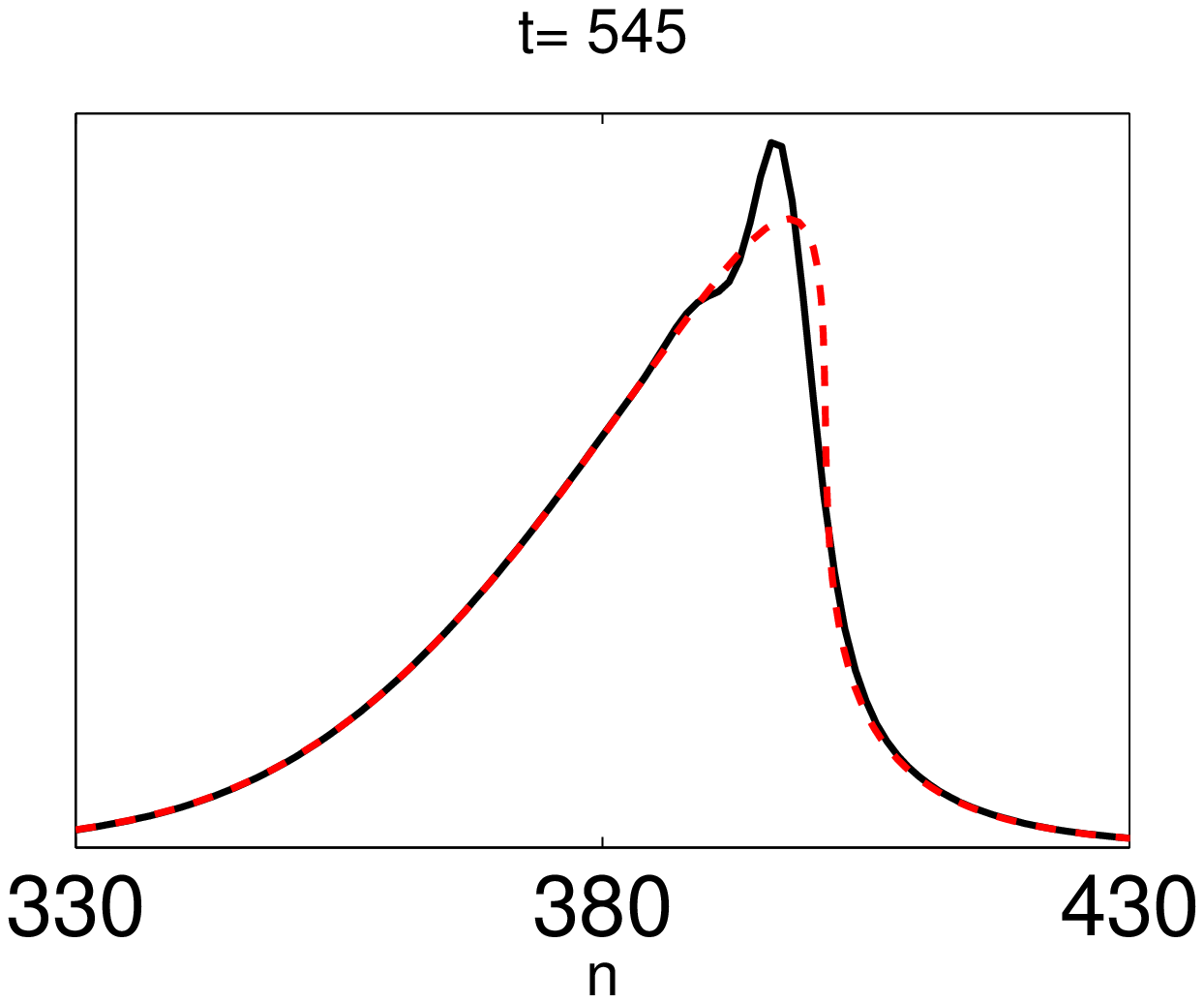,width=.33 \textwidth, height = .15 \textheight}
\epsfig{file=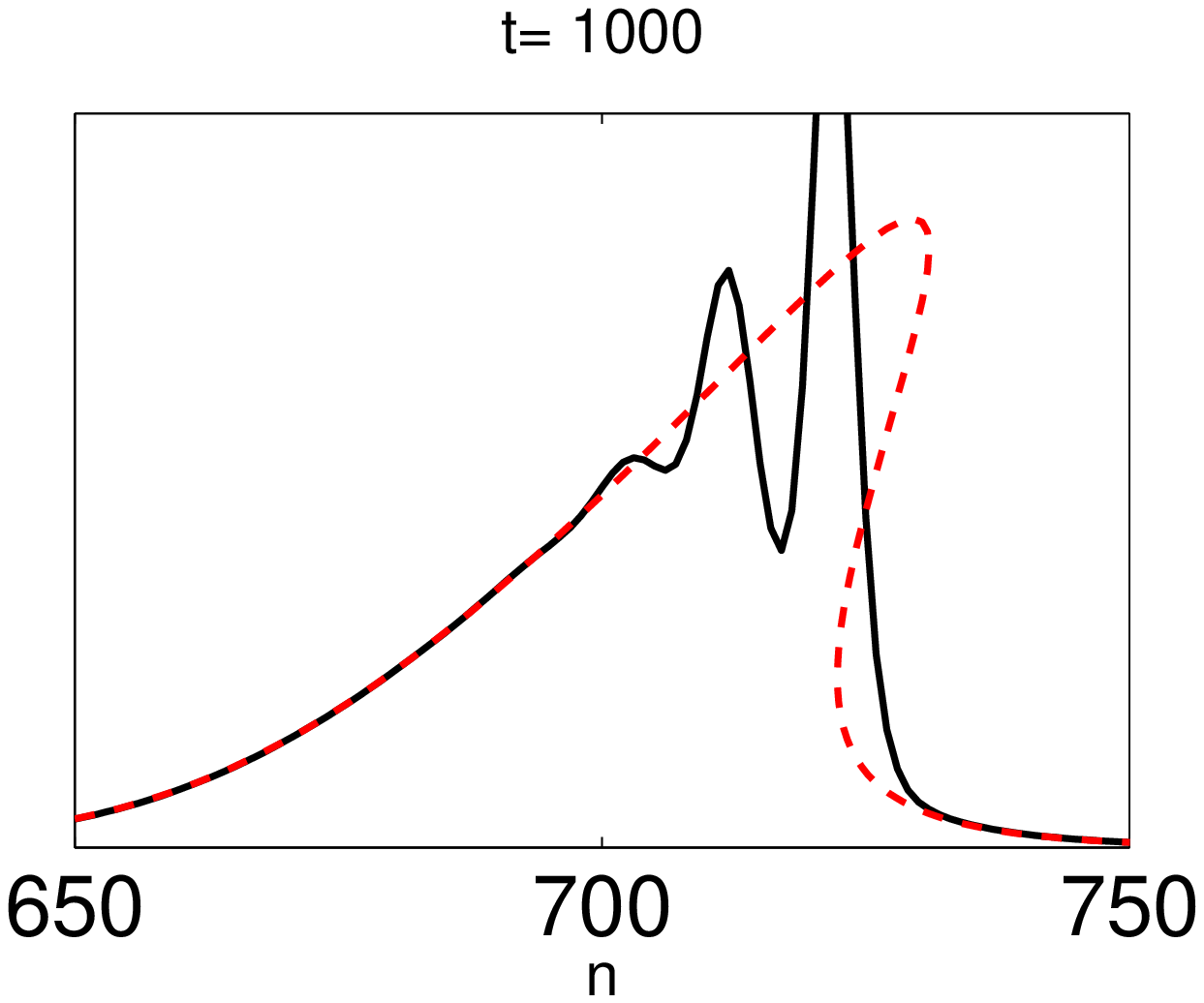,width=.33 \textwidth, height = .15 \textheight}
}
\caption{A zoom of the wave before (left) and at the approximate moment of (middle) and after (right) the point of shock development. 
The dark solid line is from the direct simulation of the discrete model, and the dashed red line is the prediction based on the 
continuum model. Notice that past the point of expected wave breaking, the lattice model develops microscopic oscillations.   The profile at $t=145$ was used to construct the red-dashed line
using the velocity relationship  $\sqrt{ p( \delta_0 + A)^{p-1}   }$, see text. Before the development of the shock,
this prediction is very accurate, see e.g. the left panel. }
 \label{fig:oscillate}
\end{figure}

\section{Justification of the KdV and NLS approximation}
\label{sec:nls_kdv}

In this section we would like to contrast the previous result with
approximation results for the KdV and NLS approximation. The major
difference lies in the ratio between the amplitude and the
precompression. For the KdV and NLS approximation this ratio is
$\mathcal{O} (\varepsilon^2)$ resp.\ $\mathcal{O} (\varepsilon)$  where $0
< \varepsilon \ll 1$ is the small perturbation parameter, whereas for
the quasicontinuum  approximation the ratio is $\mathcal{O} (1)$, i.e., of a
comparable order.
%
Since Eq.~\eqref{law2} is exactly of the form of the FPU systems considered in
\cite{CCPS12,Schn10}  the approximation theorems of
these papers also apply here. The subsequent solutions $u_n$ will be $\mathcal{O} (\varepsilon^2)$, resp.\ $\mathcal{O}
(\varepsilon)$  and hence will always live in the ball of the convergence with radius $\delta_0 $ 
if the perturbation parameter $ 0 < \varepsilon \ll 1 $ is chosen to be  sufficiently small.

Even spatially periodic arrangements can
be considered here, namely
\begin{equation} \label{eq11}
\partial^2_t u_n = a_{n+1} (u^p_{n+1}) - 2a_n (u^p_n) + a_{n-1} (u^p_{n-1})
\end{equation}
with $a_n = a_{n+N}$ for a fixed $N$. 

We formulate the relevant approximation theorems in the homogeneous case $N=1$.
For the KdV approximation we have,

\begin{theorem}
Let $A\in C(\left[ 0,T_0\right], H^6)$ be a solution of the KdV
equation $\partial_T A = \nu_1 \partial^3_{ X} A +
\nu_2 \partial_{ X} (A^2)$ with suitable chosen coefficients
$\nu_1, \nu_2 \in \mathbb{R}$.
Then there exist $\varepsilon_0 > 0, \, C>0$ such that for all
$\varepsilon \in (0, \varepsilon_0)$ we have solutions $(u_n)_{n\in
  \mathbb{Z}}$ of \eqref{eq11} with
\[
\sup_{t \in [0, T_0/\varepsilon^3]} \sup_{n \in \mathbb{N}} \left|
  u_n(t) - \psi_n (t) \right| \leq C \varepsilon^{5/2},
\]
where
\[
\psi_n (t) = \varepsilon^2 A \left( \varepsilon (n-\omega'_1 (0) t),
  \varepsilon^3 t\right) 
\]
with $\omega_1(k)^2 = \omega(k)^2 b_1$.

\end{theorem}

\begin{proof}
The proof follows trivially from \cite[Theorem 3.1]{CCPS12} or \cite{SW99FPU}.  We note there is no gap-opening in the small amplitude limit
(i.e. the Heaviside function in the potential will play no role). 
\end{proof}


The  theorem can be generalized easily to an approximation theorem for two decoupled KdV equations 
describing  counter-propagating waves, cf. \cite{SW99FPU}.

For the NLS approximation we have,
\begin{theorem}
Let $A \in C ([0, T_0], H^{19})$ be a solution of the NLS equation
$\partial_T A = i \nu_{1} \partial_X A + i \nu_2 A \left|A\right|^2$
with suitable chosen coefficients $\nu_1, \nu_2\in \mathbb{R}$.
Then there exist $\varepsilon_0 > 0, C > 0$ such that for all
$\varepsilon \in (0, \varepsilon_0)$ we have solutions $(u_n)_{n \in
\mathbb{Z}}$ of \eqref{eq11} with 
\[
\sup_{t \in [0, T_0/\varepsilon^2]} \sup_{n \in \mathbb{N}} \left|
 u_n(t) - \psi_n (t) \right| \leq C \varepsilon^{3/2}
\]
where
\[
\psi_n (t) = \varepsilon A\left( \varepsilon (n-\omega'_1 (k_0)t),
\varepsilon^2 t\right) e^{i (k_0 n-\omega_0t)} + \textup{c.c.}
\]
with $\omega_1(k)^2 = \omega(k)^2 b_1$

\end{theorem}

\begin{proof}
The proof follows trivially from  \cite[Theorem 4.1]{CCPS12} or \cite[Theorem 1.1]{Schn10}.
\end{proof}


\section{Conclusions and Future Challenges}

The main result of this paper derives in a rigorous way (and with 
controllable corrections)
the leading order quasicontinuum approximation
for long wavelength solutions in the granular 
crystal model~\eqref{law}, in accordance with the formal derivation 
of Nesterenko~\cite{Nesterenko2001}. As a technical
assumption, we required the presence of a precompression factor (while the
original Nesterenko model has been developed also in the case of the
so-called ``sonic vacuum'' i.e., without precompression). One
obvious avenue of future research is to investigate a proof without this assumption, which, however, would demand a fundamentally different 
technique than the one presented herein. 

On the other hand, perhaps an
even more important aspect of investigation
concerns the well-posedness
theory of the full Nesterenko model~\cite{Nesterenko2001} or of the variant
developed by Ahnert and Pikovsky in~\cite{pikovsky}. 
One important consequence of keeping only the leading order terms in the continuum model (as done herein) is 
the inability to capture the exact solitary wave solutions (which are known to exist in granular crystals \cite{Friesecke94,pego05}) and are at
the core of experimental observations 
in such systems~\cite{Nesterenko2001,Sen2008,Kevrekidis2011}. 
The methods of this paper cannot be directly applied to that case, although
we should note that we suspect that these higher order long
wavelength  models suffer 
(especially so in the case of precompression) from the type of 
pathologies that were identified by Rosenau and led him to devise
appropriate regularizations~\cite{r1,r2}; see 
also the more recent discussion of~\cite{titi}. It would be especially
relevant to consider such regularizations of the higher order long
wavelength models both from a rigorous, as well as from a numerical 
perspective.

Another important consequence of keeping only first order terms  in the continuum model is the prediction of the development of shock waves for a suitable
(yet broad) class of initial data.
Although the main theorem of
this paper does not apply to the case of shock waves, due to 
smoothness considerations, numerical simulations indicate the steepening
of relevant initial data towards a shock structure and suggest that this
is indeed an issue
worthy of further exploration,
with an aim towards transferring these results
to the discrete model. Indeed, it is known
in FPU lattices that this procedure fails \cite{Herrmann10}, due to the existence of high frequency oscillations (resulting
from so-called dispersion shock waves). Thus, a different continuum model 
(than the one derived herein) will most likely be needed to
fully characterize the emerging dispersive shock wave case. 
These topics are currently under consideration and will be reported
in future publications.

\section*{Acknowledgments} The authors would like to thank Alexey Miroshnikov and Robin Young for helpful discussions regarding
shock formation in the p-system.  PGK acknowledges support from the US National Science Foundation under grant
CMMI-1000337, the US Air Force under grant FA9550-12-1-0332, the Alexander von Humboldt Foundation, as well as
the Alexander S. Onassis Public Benefit Foundation.

\end{document}